\documentclass{article}

\usepackage{version}

\usepackage{amssymb}
\usepackage{amsmath}
\usepackage{txfonts}
\usepackage{amssymb}
\usepackage{enumerate}

\usepackage{times,color}
\usepackage{mathrsfs}
\usepackage{amscd}
\usepackage{stmaryrd}
\usepackage{graphicx}
\usepackage{makeidx}
\usepackage{appendix}

\usepackage{multirow}
\usepackage{fancyvrb}
\usepackage{xspace}
\usepackage{hyperref}

\newcommand{\qed}{$\Box$}
\newenvironment{proof}{\noindent {\bf Proof:}}{\hfill \qed\\}

\newcommand{\Agt}{\mathit{Agt}}

\newcommand{\WMLOKP}{\mathbf{WMLOKP}} 
\newcommand{\WMLOP}{\mathbf{PMLO}} 
\newcommand{\commentout}[1]{}

\newcommand{\I}{\mathcal{I}}
\newcommand{\R}{\mathcal{R}}
\newcommand{\K}{\mathcal{K}}
\newcommand{\cS}{\mathcal{S}}
\newcommand{\Pspace}{\mathbf{Pr}}
\newcommand{\F}{\mathcal{F}}
\newcommand{\G}{\mathcal{G}}
\newcommand{\Nat}{\mathbb{N}}

\newcommand{\Rat}{\mathbb{Q}}
\newcommand{\Real}{\mathbb{R}}

\newcommand{\PLTLsK}{\mbox{$\mathbf{CTL^*KP}$}\xspace }
\newcommand{\PCTLK}{\mbox{$\mathbf{CTLPK}$}\xspace}
\newcommand{\CTL}{\mbox{$\mathbf{CTL}$}\xspace}

\newcommand{\infinitepaths}[1]{\mathtt{P}_\infty(#1)}
\newcommand{\finitepaths}[2]{\mathtt{P}_{#2}(#1)}
\newcommand{\pathsover}[1]{\uparrow #1}

\newcommand{\rimp}{\Rightarrow}
\newcommand{\dimp}{\Leftrightarrow}

\newcommand{\prob}{{\tt Pr}}
\newcommand{\prior}{\mathtt{Prior}}
\newcommand{\until}{U} 

\newcommand{\powerset}[1]{{\cal P}(#1)}

\newcommand{\Prop}{\mathit{Prop}}
 
\def\Prt[#1]{Pr^{(#1)}}

\newcommand{\kset}{\mathcal{K}}
 
\newcommand{\spr}{\mathtt{spr}}
\newcommand{\clk}{\mathtt{clk}}

\newcommand{\ve}[1]{{\mathbf #1}}

\newcommand{\beq}{\begin{equation}}
\newcommand{\eeq}{\end{equation}}
\newcommand{\be}{\begin{enumerate}}
\newcommand{\ee}{\end{enumerate}}

\newtheorem{theorem}{Theorem} 
 
\newtheorem{lemma}[theorem]{Lemma} 
\newtheorem{propn}[theorem]{Proposition}

\newcommand{\pr}{\operatorname{pr}}

\newcommand{\PAS}{Q}
\newcommand{\PAalph}{\Sigma} 
\newcommand{\PAinit}{\mathbf{v}_0} 
\newcommand{\PAtrans}{A} 
\newcommand{\PAfinal}{F}
\newcommand{\PFA}{{\cal A}}

\newcommand{\timeassgt}{\tau}

\begin{document}

\title{Undecidable Cases of Model Checking  Probabilistic Temporal-Epistemic Logic\thanks{Version of September 15, 2015. This is an extended version, with full proofs, of
a paper  that appeared in TARK 2015.
It corrects an error in the TARK 2015 pre-proceedings version, in the definition of mixed-time polynomial atomic probability formulas.} 
} 

\author{
Ron van der Meyden\\
Manas K Patra\\[5pt]
School of Computer Science and Engineering\\
UNSW Australia
}
\date{}

\maketitle

\begin{abstract} 
We investigate the decidability of model-checking logics of time, knowledge and probability, with respect to two epistemic semantics:
the clock and synchronous perfect recall semantics in partially observed discrete-time Markov chains. Decidability results are known for certain restricted logics with respect to these semantics, 
subject to a variety of restrictions that are either unexplained or  involve a longstanding unsolved mathematical problem. 
We show that mild generalizations of  the known decidable cases suffice to render the model checking
problem definitively undecidable. In particular, for a synchronous perfect recall, a generalization from temporal operators with finite reach 
to operators with infinite reach renders model checking undecidable. The case of the clock semantics is 
closely related to a monadic second order logic of time and probability that is known to be  decidable, except on a set of measure zero.
We show that two distinct extensions of this logic make model checking undecidable. One of these involves polynomial combinations of 
probability terms, the other involves monadic second order quantification into the scope of probability operators. These results explain
some of the restrictions in previous work. 
\end{abstract}

\section{Introduction} 

{\em Model checking} is a verification methodology used in computer science, 
in which we  ask whether a given model satisfies a given formula of some logic.  
First proposed in the 1980's \cite{CE81}, 
model checking is now a rich area, with a large body of associated theory and 
well developed implementations that automate the task of model checking. 
Significant use of model checking tools is made in industry, in particular, 
in the verification of computer hardware designs. 

Model checking developed originally in a setting where the 
specifications are expressed in a propositional temporal logic, and
the systems to be verified are  finite state automata. This setting has
the advantage of being decidable, and a great deal of work has gone into
the development of algorithms and heuristics for its efficient implementation. 
More recently,  the field has explored the extent to which the expressiveness 
of both the model representations and of the specification language can be 
extended while retaining decidability of model checking. 
Extensions in the systems dimensions considered include real-time systems \cite{AlurCD90},  systems
with a mixed continuous and discrete dynamic \cite{MalerNP08}, richer automaton models
such as push-down automata, machines with first-in-first out queues etc. 
In the dimension of the specification language, extensions considered 
include elements of second order logic and specific constructs
to capture the richer properties of the systems models described above
(e.g. in the real time case the specification language might
contain inequalities over time values.)  

Model checking for epistemic logic was first mooted in 
\cite{HV91},  and model checking for the combination 
of temporal and epistemic logic has been developed
both theoretically  \cite{MeydenShilov,EGM07,HM10} and in practice \cite{mck,MCMAS,Verics,DEMO}.  
A variety of semantics for knowledge are known to be associated with 
decidable model checking problems in finite state systems, in particular, the 
observational semantics (in which an agent reasons based on its present observation) 
the clock semantics (in which an agent reasons based on its present observation and the present clock value), 
and synchronous and asynchronous versions of perfect recall, all admit
decidable model checking in combination with quite rich temporal expressiveness \cite{MeydenShilov,EGM07,HM10}.

Orthogonally, a line of work on probabilistic model checking
has considered model checking of assertions about
probability and time \cite{prismbook}.  
Although one might at first
expect this line of work to be closely related to epistemic
model  checking, in that probability theory provides 
a model of uncertainty, in fact this area has been concerned not with 
how subjective probabilities change over time, but  with a probabilistic extension of temporal logic. 
The focus tends to be on the prior probability of some temporal property,
or on the probability that some temporal property holds in
runs from a current {\em known} state. 

Rather less attention has been given to model checking the combination 
of subjective probability and temporal expressiveness. 
Of the semantics for knowledge mentioned above, the 
clock and synchronous perfect recall semantics
are most suited as a basis for model checking subjective probability. 
(The others suffer from asynchrony, which 
makes it more difficult  to associate  a single natural 
probability space.) Implementations for these semantics
presently exist only for a limited set of formulas, in which the 
full power of temporal logic is not used. For example, results in \cite{HLM11} for model checking the logic of 
subjective probability (with clock or synchronous perfect recall semantics) 
and time restricts the temporal 
operators to have only finite reach into the future, 
and does not handle operators such as ``at all times in the future". 

A fundamental reason underlying this is that the problem of model checking probability with a rich temporal 
expressiveness  seems to be inherently complex. 
Indeed, it requires a solution to a basic mathematical problem, the {\em Skolem Problem}  
for linear recurrences, that has stood unsolved since first posed in the 1930's \cite{Skolem34}. 
Consequently, the strongest results on model checking
probability and time that encompass the expressiveness required 
for model checking knowledge and subjective probability  state
decidability in a way that requires exclusion of an infinite set of 
difficult instances for which decidability is unresolved. 
Specifically, \cite{BRS06} shows that a (weak) monadic second order logic $\WMLOP$, containing
probability assertions of forms such as $\prob(\phi(t_1, \ldots, t_n)) > c$, in which the
$t_i$ take values in the natural numbers, representing discrete time points, 
is decidable in finite state Markov chains, provided that the rational number $c$ is not in 
a set $H_\phi$ depending on $\phi$ which can taken to be of arbitrarily small non-zero measure. This work leaves open the decidability of the model checking 
problem for the language in its full generality, in particular, for the  values of $c$ in $H_\phi$.  

Our contribution  in this paper is to consider a number of generalizations of $\WMLOP$, motivated by model checking a logic of time and subjective probability. 
In particular, our generalizations arise very naturally when attempting to deal with the way that an agent 
conditions probability on its observations.  We show that these generalizations definitively result in 
undecidable model checking problems.  This clarifies the boundary between the decidable and undecidable cases of 
model checking  logics of probability and time. 

We begin in section~\ref{sec:lang} by recalling the definition of {\em probabilistic interpreted systems} \cite{HalpernUncertainty}, 
which provides a very general semantic framework for logics of time, knowledge and probability. We work with an instantiation of
this general framework in which systems are generated from finite state partially observed discrete-time Markov chains. 
We define two logics that take semantics in this framework. 
The first is an extension of the branching time temporal logic $\CTL^*$
to include operators for knowledge and probability, including operators for the subjective probability of agents. 
The second is a more expressive monadic second order  logic that  also adds a capability to quantify over moments of time
and \emph{finite sets} of moments of time. In this logic, the agent knowledge and probability operators
are indexed by a temporal variable. This logic generalizes the logic of \cite{BRS06}. 
Our logics allow polynomial comparisons of probability terms, as well as comparisons of agent probability terms referring to multiple time points. 
We argue from a number of motivating applications that this level of expressiveness is
useful in potential applications.  
We show in Section~\ref{sec:relations} that the monadic second order logic is as least as expressive as our probabilistic extension of $\CTL^*$. 
Indeed, some apparently mild extensions of $\WMLOP$ suffice for the encoding:  
the epistemic  and subjective probability operators can be eliminated using a universal modality,  
polynomial combinations of probability expressions, and a more liberal use of 
quantification than allowed in $\WMLOP$. 

We then turn in Section~\ref{sec:results} to an investigation of the model checking problem. 
Specifically, we show that model checking even very simple formulas about a single agent's probability 
is undecidable when the agent has perfect recall. A consequence of this result is that 
an extension of $\WMLOP$  that adds second order quantification into the scope 
of probability is undecidable. 

This suggests a focus on weaker epistemic semantics instead, in particular, the 
clock semantics. From the point of view of $\WMLOP$, to express
agent's subjective probabilities with respect to the clock semantics requires polynomial combinations of 
simple global probability terms of the form `` the probability that proposition $p$ holds at time $t$". 
We formulate a simple class of formulas involving such  polynomial combinations, 
and show that this  also has undecidable model checking.  

These results show that even simple model checking questions about 
subjective probability are undecidable, and moreover help to 
explain some unexplained restrictions on $\WMLOP$ in \cite{BRS06}: 
these restrictions are in fact necessary in order to obtain a decidable logic.  
We conclude with a discussion of future work
in Section~\ref{sec:concl}. Related work most closely related to our results is discussed in the 
context of presenting and motivating the results.

\section{Probabilistic Knowledge}\label{sec:lang} 

We describe in this section the semantic  setting 
for the model checking problem we consider. 
We model a set of agents making partial observations of an environment 
that evolves with time. We first present the semantics of the modal logic we consider, following \cite{HalpernUncertainty}, 
using the general notion of probabilistic interpreted system. 
Since these structures are not finite, in order to have a finite input for a model checking problem, 
we derive a probabilistic interpreted system from a partially observed discrete-time Markov chain. 
This is done in two ways, depending on the degree of recall of the agents. Taking the 
Markov chain to be finite, we obtain finitely presented model checking problems whose complexity we
then study. 

\subsection{Probabilistic Interpreted Systems} 

Probabilistic interpreted systems are defined as follows. 
Let  $Agt=\{1,\ldots,n\}$ be a set of agents operating in an environment $e$. 
At each moment of  time, each  agent is assumed to be in some {\it local} state, which records all the information that the agent can access at that time. 
The environment $e$ records ``everything else that is relevant". 
Let $S$ be the set of environment states and let $L_i$ be the set of local states of agent $i\in Agt$. 
A {\it global} state of a multi-agent system is an $(n+1)$-tuple $s=(s_e,s_1,\ldots,s_n)$ such that $s_e\in S$ and $s_i\in L_i$ for all $i\in Agt$.
We write $\G = S \times L_1 \times \ldots \times L_n$ for the set of global states. 

Time is represented discretely using the natural numbers $\mathbb{N}$. 
A {\em run}  is a function $r:\mathbb{N}\rightarrow \G$, specifying a 
global state at each moment of time. A pair $(r,m)$ consisting of a run $r$ and time $m\in \Nat$ is called a {\em point}. If $r(m)=(s_e,s_1,\ldots,s_n)$ then 
we define $r_e(m)=s_e$ and $r_i(m)=s_i$ for $i\in Agt$.  
If $r$ is a run and $m\in \Nat$ a time, we write $r[0..m]$ for $r(0) \ldots r(m)$ and $r_e[0..m]$ for $r_e(0) \ldots r_e(m)$. 
A {\em system} is a set $\R$ of runs. 
We call $\R\times \Nat$  the {\em set of points} of the system $\R$. 

Agent knowledge is captured using a relation of indistinguishability. 
Two points $(r,m)$ and $(r',m')$ are said to be \emph{indistinguishable to agent $i$}, 
if the agent is in the same local state at these points. 
Formally, we define $\sim_i$ to be the equivalence relation on $\R\times \Nat$ given by $(r,m) \sim_i (r',m')$, if $r_i(m) = r'_i(m')$. 
Relative to a system $\R$, we define 
the set  $$\kset_i(r,m)=\{(r',m')\in\R\times \mathbb{N}~|~(r',m')\sim_i(r,m)\}$$ to be the set of points that are,  for agent $i$,  indistinguishable from  the 
point $(r,m)$. 
Intuitively, $\kset_i(r,m)$ is the set of all points that the agent considers possible when it is in the actual situation $(r,m)$. 
A system is said to be \emph{synchronous} if for all agents $i$, we have that $(r',m') \in \kset_i(r,m)$ implies that $m=m'$. 
Intuitively, in a synchronous system, agents always know the time. Since it is more
difficult to define probabilistic  knowledge in systems that are not synchronous, we confine our attention to synchronous systems
in what follows.

A {\em probability space} is a triple $\Pspace=(W,\F,\mu)$ such that $W$ is a (nonempty) set, 
called the {\em carrier}, $\F\subseteq \mathcal{P}(W)$ is a $\sigma$-field of subsets of $W$, called the {\em measurable} sets in $\Pspace$,
containing $W$ and 
closed under complementation and countable union, 
and $\mu:\F\rightarrow [0,1]$ is a {\em probability measure},
such that $\mu(W)=1$ and $\mu(\bigcup_n V_n)=\sum_n \mu(V_n)$ for every countable sequence $\{V_n\}$ of mutually disjoint measurable sets $V_n \in \F$.  
As usual, we define the conditional probability $\mu(U|V) = \mu(U\cap V)/\mu(V)$ when $\mu(V) >0$. 

Let $\Prop$ be a set of \emph{atomic propositions}. 
A {\em probabilistic interpreted system} over $\Prop$ is a tuple $\I=(\R, \prob_1, \ldots , \prob_n, \pi)$ 
such that $\R$ is a system, each $\prob_i$ is a function mapping each point $(r,m)$ of 
$\R$ to a probability space  $\prob_i(r,m)$ in which the carrier is a subset of $\R\times \Nat$, 
and $\pi:\R\times \mathbb{N}\rightarrow \powerset{\Prop}$ is an interpretation of some set $\Prop$ of 
atomic propositions. Intuitively, the probability space $\prob_i(r,m)$ captures the way that the agent $i$ assigns
probabilities at the point $(r,m)$, and $\pi(r,m)$ is the set of atomic propositions that are true at the point. 

We will work with probabilistic interpreted systems derived from 
synchronous systems in which agents have a common prior on the set of runs.
To define these, we use the following notation. 
For a system $\R$, a set of runs $\cS \subseteq \R$ and a set of points $U\subseteq \R\times \Nat$, define 
$$\cS(U)=\{r\in\cS~|~\exists m:(r,m)\in U\}$$ to be the set of runs in $\cS$ 
passing through 
some point in
the set $U$.  Conversely, for a set $\cS$ of runs and a set $U$ of points, 
define $$U(\cS) = \{(r,m)\in U~|~r\in \cS\}$$ to be the set of points in $U$ 
that are on a run in $\cS$. 
Note that if there exists a constant $k\in \Nat$ such that $(r,m) \in U$ implies $m=k$, 
then the relation $r \leftrightarrow (r,k)$ defines a  one-to-one correspondence between $\cS(U)$ and $U(\cS)$. 
In synchronous systems, 
in which the sets $\kset_i(r,m)$ satisfy this condition, 
this gives a way to move between sets of 
points considered possible by an agent and corresponding sets of runs. 

Suppose that $\R$ is a synchronous system,  let $\Pspace = (\R,\F,\mu)$ be a probability space on the system $\R$, 
and let $\pi$ be an interpretation on $\R$.  Intuitively, the probability space $\Pspace$ represents a 
prior distribution over the runs. We assume that for all points $(r,m)\in \R\times \Nat$ and agents $i$, 
we have that $\R(\kset_i(r,m))\in \F$ is a measurable set and $\mu(\R(\kset_i(r,m))) >0$.   
(This assumption can be understood as saying that, according to the prior, each 
possible local state $r_i(m)$ of agent $i$ at time $m$ has non-zero probability of being the local 
state of agent $i$ at time  $m$.)  
Under this condition, we define the probabilistic interpreted system 
$\I(\R,\Pspace,\pi) = (\R,\prob_1,\ldots,\prob_n,\pi)$ such that $\prob_i$ associates with each point $(r,m)$ the probability space
 $\prob_i(r,m)=(\kset_i(r,m),\F_{r,m,i},\mu_{r,m,i})$
 defined by $$\F_{r,m,i} = \{\kset_i(r,m)(\cS)~|~\cS \in \F\}$$ and  
 such that $$\mu_{r,m,i}(U)=\mu(\, \R(U)~|~\R(\kset_i(r,m))\, )$$ for all $U \in \F_{r,m,i}$. 
Intuitively, because the set of runs $\R(\kset_i(r,m))$ is measurable, we can obtain a probability space with carrier
$\R(\kset_i(r,m))$ by conditioning in $\Pspace$. Because of the synchrony assumption there is, 
for each point $(r,m)$, a one-to-one correspondence 
between points in $\kset_i(r,m)$ and runs in $\R(\kset_i(r,m))$. The construction 
uses this correspondence to induce a probability space on $\kset_i(r,m)$ from the probability space on 
$\R(\kset_i(r,m))$.
We remark that under the additional assumption of perfect recall, it is 
also possible to understand each space $\prob_i(r,m+1)$ as 
obtained by conditioning on the space $\prob_i(r,m)$. 
See \cite{HalpernUncertainty} for a detailed explanation of this point.

\subsection{Probabilistic Temporal Epistemic Logic} 

To specify properties of probabilistic interpreted systems, a variety of logics can be formulated, 
drawing from the spectrum of temporal logics.  Our main interest is in a reasoning about subjective probability and
time, so we first consider a natural way to combine existing temporal and probabilistic 
logics. For purposes of comparison, it is also helpful to consider a rather richer monadic second 
order logic of probability and time, that is closely related to a logic for which some decidability results are known. 

We may combine temporal and probabilistic logics to define a logic \PLTLsK that extends the temporal logic $\mathbf{CTL^*}$ 
by adding operators for knowledge and probability.
Its syntax is given by the grammar
\[
\begin{array}{l} 
\phi~::=~p~|~\neg\phi~|~\phi\wedge \phi~|~A\phi~|~X\phi~|~\phi \until \phi~|~K_i\phi~|~
f(P, \ldots,P) \bowtie c\\[8pt] 
P ::= \prob_i(\phi) ~|~\prior_i(\phi) 
\end{array} 
\]
where $p\in \Prop$, 
$c$ is a rational constant,  
$\bowtie$ is a relation symbol in the set $\{\leq,<,=,>,\geq\}$, 
and $f(x_1, \ldots, x_k)$ is multivariate polynomial in $k$ variables 
$x_1 , \ldots x_k$ with rational coefficients. 
Instances of $P$ are called \emph{basic probability expression}. 
The instances generated from  $f(P, \ldots,P)$ are 
called {\em probability expressions}, and are expressions of the form 
 $f(P_1, \ldots , P_k)$, obtained by substituting a basic probability expression $P_i$ for each variable $x_i$ in $f(x_1, \ldots, x_k)$.  
For example, 
$$4  \prob_1(p)^5\cdot \prob_2(q)^3 + \frac{7}{15}\prob_1(p) $$ 
is an instance of $f(P, \ldots,P)$ 
obtained from 
$f(x,y) = 4x^5y^3 + \frac{7}{15}x$ by substituting $\prob_1(p)$ for $x$ and $\prob_2(q)$ for $y$.

 Intuitively, formula $K_i\phi$ expresses that agent $i$ knows $\phi$. 
The formula $A \phi$ says that $\phi$ holds for all possible system evolutions from the current situation. 
The formula $X\phi$ expresses that $\phi$ holds at the next moment of time. 
The formula $\phi_1 \until \phi_2$ says that $\phi_2$ eventually holds, and  $\phi_1$ holds until that time. The expression 
 $\prob_i(\phi) $ represents agent $i$'s current probability of $\phi$, 
  $\prior_i(\phi) $ represents agent $i$'s {\em prior} probability of $\phi$, i.e., the agent's probability of 
  $\phi$ at time 0. 
 The formula $f(P_1, \ldots,P_k) \bowtie c$ expresses that this polynomial combination of 
 current and prior probabilities stands in the relation $\bowtie$ to $c$. 
 We use standard abbreviations from temporal logic, in particular, 
 we write $F\phi$ for $\mathit{true} \until \phi$.
 
 A restricted fragment of the language that may be of interest 
 is the {\em branching time fragment}  in which the temporal operators 
 are restricted to those of  the temporal logic $\mathbf{CTL}$. 
 That is,  $X$ and $\until$ are permitted to 
 occur only in combination with the operator $A$, in one of the forms 
 $AX\phi$, $EX\phi$,  $A\phi_1 \until \phi_2$, $E\phi_1\until\phi_2$, 
 where we write $E\phi$ as an abbreviation for $\neg A \neg \phi$. 
 We call this fragment of the language \PCTLK. 
 The motivation for considering this fragment is that 
 the complexity of model checking   is in polynomial time
 for  the temporal logic  $\CTL$, whereas it is 
 polynomial-space complete for the richer temporal logic $\CTL^*$ \cite{CES86}. 
The logic \PCTLK is therefore, {\em prima facie}, a candidate 
for lower complexity once knowledge and probability operators
are added to the logic.

 The semantics of the language \PLTLsK in a probabilistic interpreted system 
 $\I = \I(\R, \Pspace, \pi)$ is given by interpreting formulas $\phi$ at points $(r,m)$ of $\I$, 
 using  a satisfaction relation $\I,(r,m)\models \phi$. 
The definition is mutually recursive with a function $[\cdot]_{\I,(r,m)}$
that assigns a value $[P]_{\I,(r,m)}$  to each probability expression $P$ at each point $(r,m)$. 
This requires computing the measure of certain sets. 
For the moment, we assume that all sets arising in the definition are measurable. 
We show later that this assumption holds in the cases of interest in this paper. 

 We first interpret the probability expressions at points $(r,m)$ of the system $\I$, by 
 
\begin{subequations}
\begin{equation*}\label{eq:ProbAssignment_agent}
[\prob_i\phi]_{\I,(r,m)}   =  \mu_{r,m,i}(\{(r',m')\in \kset_i(r,m)~|~\I,(r',m')\models\phi\})
\end{equation*}
\begin{equation*}
 ~[\prior_i\phi]_{\I,(r,m)}   =   \mu_{r,0,i}(\{(r',0)\in \kset_i(r,0)~|~ \I,(r',0)\models\phi\})
\end{equation*}
\begin{equation*}
 ~[f(P_1, \ldots,P_k)]_{\I,(r,m)}   =   f([P_1]_{\I,(r,m)} , \ldots,[P_k]_{\I,(r,m)} )
\end{equation*}

\end{subequations}

 The satisfaction relation is then defined recursively, as follows: 
\begin{enumerate}
\item $\I, (r,m)\models p$ if $p\in\pi(r,m)$
\item $\I, (r,m)\models \neg\phi$ iff not $\I,(r,m)\models \phi$
\item $\I, (r,m)\models \phi_1\wedge\phi_2$ iff $\I, (r,m)\models \phi_1$ and $\I, (r,m)\models \phi_2$
\item $\I, (r,m)\models A\phi$ if $\I,(r',m)\models \phi$ for all runs $r'$ with $r'[0\ldots m] = r[0\ldots m]$, 
\item $\I, (r,m)\models X\phi$ if $\I,(r,m+1)\models \phi$
\item $\I, (r,m)\models \phi_1 \until \phi_2$ holds if there exists $k\geq m$ such that 
$\I,(r,k)\models \phi_2$, and $\I,(r,l)\models \phi_1$ for all $l$ with $m\leq l<k$. 
\item $\I, (r,m)\models K_i\phi$ if $\I, (r',m')\models \phi$ for all $(r',m')\in \kset_i(r,m)$.
\item $\I, (r,m)\models f(P_1, ...,P_k) \bowtie c$ if $[f(P_1, ...,P_k)]_{\I,(r,m)}  \bowtie c$. 
\end{enumerate}

\subsection{Probabilistic Monadic Second Order Logic} 

\newcommand{\botag}{\bot} 
\newcommand{\topag}{\top} 

Temporal modal logics refer to time in a somewhat implicit way. An alternative
approach is to work in a setting with more explicit references to time, 
by using variables denoting time points.  Kamp's theorem \cite{kamp} establishes an equivalence in the first order case, but
by adding second order variables and quantification, one
can obtain richer logics, that frequently remain decidable in the 
monadic case.  In this section, we develop a logic in this style for
time and subjective probability. 

We define the logic $\WMLOKP$ as follows. 
We use two types of variables: time variables $t$ and set variables $X$. 
Time variables take values in $\Nat$ and set variables take \emph{finite} subsets of $\Nat $ as values. 
{\em Probability terms}  $P$ have  the form  $\prob(\phi)$ or the form 
$ \prob_{i,t}(\phi)$ where $i \in \Agt$ is an agent, $t$ is a time variable, $\phi$ is a formula. 
Formulas $\phi$ are defined by the following grammar: 
$$ \phi ::= \begin{array}[t]{l} 
p(t) ~ |~X(t) ~|~ t_1 < t_2~|~f(P, \ldots, P) \bowtie c   ~|~ \neg \phi~|~\phi \land \phi~|~ \\
K_{i,t}( \phi) ~|~\forall t(\phi) ~|~ \forall X(\phi)
\end{array} 
$$ 
where $t, t_1, t_2$ are time variables, $p$ is an atomic  proposition, $X$ is a set variable,  $i$ is an agent, $c \in \Rat$ is  a rational constant, 
$f$  is a rational polynomial (see the discussion above for $\PLTLsK$), 
and $\bowtie$ is a relation symbol from the set $\{ =, < ,\leq, > ,\geq\}$. 

Intuitively, in this logic formulas are interpreted relative to a run. 
Instead of indexing by a single moment of time, as in the logic above, we 
relativize the satisfaction relation to an assignment of values to the temporal and set variables. 
Atomic formula  $p(t)$ says that proposition  $p$ holds at time $t$.
Similarly, a (finite) set $X$ of times can be interpreted as a proposition, and
we can understand  $X(t)$ as stating that the value of $t$ is in $X$. 
(We remark that there is a fundamental difference between the types of 
propositions denoted by  atomic propositions $p$ and set variables $X$: whereas the 
atomic propositions may depend on structural aspects of the run, such as
the global state at time $t$, the set variables may refer only to the time.)  
The atomic formula $t_1 < t_2$ has the obvious interpretation that time $t_1$ is less than time $t_2$. 
The constructs
$\forall t(\phi)$  and $\forall X(\phi)$ correspond to universal quantification over
times and \emph{finite}
sets of times respectively. They say that $\phi$ holds on the current run for all values of the 
variable. (Taking finite sets amounts to the \emph{weak} interpretation of second order quantification. 
One could also consider a strong semantics allowing infinite sets of times. We have opted
here for the weak interpretation to more easily relate our results to the existing literature.) 

The probability term $\prob(\phi)$ refers to the probability of $\phi$ in the probability space on runs. 
The meaning of probability term $ \prob_{i,t}(\phi)$ is agent $i$'s probability at time $t$ that the run satisfies $\phi$. 
Similarly, $K_{i,t}\phi$ says that agent $i$ knows at time $t$ that the run satisfies $\phi$. 
Note that, whereas in $\PLTLsK$, the formula $K_i \phi$ always expresses that agent $i$ knows
that $\phi$ holds at the ``current time", in  $\WMLOKP$, formulas such as 
$$ \exists u (u< t \land K_{i,t}( p(u)))$$ 
talk about the agent's knowledge, at some time $t$, about
what was true at some earlier time $u$. A similar point applies to
probability expressions.

Accordingly, for the semantics of $\WMLOKP$, we use a variant of interpreted
systems in the form $\I = (\R, \Pr, \pi)$, where $\R$ is a system, i.e., a set of runs, 
and $\pi$ is an interpretation, as above, but where 
$\Pr = (\R, \F, \mu)$ is a probability space with carrier equal to the set of runs 
$\R$, rather than a mapping associating a probability space over a set of points
with  each  agent at each point. 

When dealing with formulas with free time and set variables, we need the extra notion of an assignment for the 
time and set variables. This is a function $\timeassgt$ such that for each free time variable $t$ we have $\timeassgt(t) \in \Nat$, 
and for each free set variable $X$ 
we have that $\timeassgt(X)$ is a finite subset of  $\Nat$.
Given such an assignment, we give the semantics of probability terms and 
formulas by a mutual recursion. We give the semantics of formulas $\phi$
by means of a relation $\I, \timeassgt, r \models \phi$ defined as follows: 
\begin{enumerate}
\item $\I, \timeassgt, r\models p(t)$ if $p\in\pi(r,\timeassgt(t))$, when $p $ is an atomic proposition, 
\item $\I, \timeassgt, r\models X(t)$ iff $\timeassgt(t) \in \timeassgt(X)$, if $X$ is a set variable, 
\item  $\I, \timeassgt, r\models t_1<t_2 $ iff $\timeassgt(t_1) < \timeassgt(t_2)$, 
\item $\I, \timeassgt, r\models \neg \phi$ iff not  $\I, \timeassgt, r\models \phi$, 
\item $\I, \timeassgt, r\models \phi_1\wedge\phi_2$ iff $\I, \timeassgt, r\models \phi_1$ and $\I, \timeassgt, r\models \phi_2$,
\item $\I, \timeassgt, r\models K_{i,t}(\phi)$ if $\I, \timeassgt, r'\models \phi$ for all $(r',m')\in \kset_i(r,\timeassgt(t))$, 

\item $\I, \timeassgt, r\models f(P_1, ...,P_k) \bowtie c$ if $[f(P_1, ...,P_k)]_{\I,\timeassgt,r} \bowtie c$, 
\item $\I, \timeassgt, r\models \forall t (\phi)$ if $\I, \timeassgt[t\mapsto n], r\models \phi$ for all $n \in \Nat$,  
\item $\I, \timeassgt, r\models \forall X (\phi)$ if $\I, \timeassgt[X\mapsto U], r\models \phi$ for all finite $U \subseteq \Nat$.  
\end{enumerate}
In item (7), 
the definition is mutually recursive with the semantics of 
 probability  terms, which are interpreted as real numbers, relative to a temporal assignment. 
We define 
$$ [\prob(\phi) ]_{\I,\tau,r} = \mu( \{ r' \in \R ~|~  \I,\timeassgt,r' \models \phi\})$$  and 
$$ [\prob_{i,t}(\phi) ]_{\I,\tau,r} = \frac{\mu( \{ r'\in \R ~|~ (r, \timeassgt(t)) \sim_i (r', \timeassgt(t)),~ \I,\timeassgt,r' \models \phi\})}{\mu( \{ r'\in \R ~|~ (r, \timeassgt(t)) \sim_i (r', \timeassgt(t)) \})} $$ 
$$ [f(P_1, \ldots, P_k)]_{\I, \tau,r} = f([P_1]_{\I, \tau,r}, \ldots, [P_k]_{\I, \tau,r}) $$ 
As above, we assume measurability of the sets required, and 
also that the agent probability expressions do not involve a division by zero. We
later justify that this holds in the 
particular setting of interest in this paper. 

A particular class of formulas of $\WMLOKP$ will be of interest below. 
Define a {\em mixed-time polynomial atomic probability formula} 
to be a formula of the form%
\footnote{The TARK 2015 pre-proceedings version of this paper incorrectly had a universal 
quantifier in this definition. The existential form is needed for the correctness of Theorem~\ref{thm:mixedtime}.}
$$ \exists t_1 \ldots t_n ( f(\prob(\phi_1),  \ldots , \prob(\phi_m)) =0 )$$
where $f(x_1, \ldots, x_m)$ is a rational polynomial and each $\phi_i$ is an atomic  formula of the 
form $p(t_j)$ for some proposition $p$ and $j\in \{1\ldots n\}$.  
We motivate the usefulness of such temporal mixing of probability expressions in Section~\ref{sec:discussion}.  

The logic $\WMLOKP$ generalizes several logics from the literature. 
If we restrict the language by excluding the probability comparison atoms 
$f(P_1, \ldots, P_k) \bowtie c $ 
and knowledge formulas  $K_{i,t}( \phi)$, we have the \emph{Weak Monadic Logic of Order}, 
which is equivalent to WS1S \cite{Buchi60}. 
We obtain the \emph{Probabilistic Monadic Logic of Order} considered in \cite{BRS06}, which we denote here by $\WMLOP$, 
if we 
\begin{itemize} 
\item exclude the knowledge operators $K_{i,t}$, 
\item exclude agent's probability terms $\prob_{i,t}(\phi)$, and 
\item limit the global probability  comparisons to be of the form $\prob(\phi(t_1, \ldots, t_k)) \bowtie c $, containing
just a single probability term $\prob(\phi(t_1, \ldots, t_k))$, with the further constraint that the only free variables of $\phi$ should be 
temporal variables $t_1, \ldots t_k$. 
\end{itemize}
In particular, second-order quantification into probability expressions, e.g.,  $\forall X[ \prob(X(t)) >c ]$ is not 
permitted in $\WMLOP$,  but second order quantification that does not cross a probability operator, such as  $ \prob(\forall X[X(t))]) >c$, is allowed.
We note that $\WMLOP$ \emph{does} allow first order quantifications into the scope of probability, 
such as $\forall t[ \prob(p(t)) >c ]$. 

In the sequel, we refer to quantification into the scope of a knowledge formula or probability expression as \emph{quantifying-in}. 

\subsection{Partially Observed Markov Chains} 

Although they provide a coherent semantic framework, probabilistic interpreted systems are infinite structures, and therefore 
not suitable as input for a model checking algorithm. We therefore work with  a type of finite model called an {\em interpreted
partially observed discrete-time Markov chain}, or PO-DTMC for short. A finite PO-DTMC for $n$ agents 
is a tuple  $M=(S,PI,PT,O_1,...,O_n,\pi)$, where $S$ is a finite set of states, $PI:S\rightarrow[0..1]$ is a function such that $\sum_{s\in S}PI(s)=1$, component $PT:S\times S\rightarrow [0,1]$ is a function 
such that $\sum_{s'\in S}PT(s,s')=1$ for all $s\in S$,  and for each agent $i\in Agt$, we have a function $O_i:S\rightarrow \mathcal{O}$ for some set $\mathcal{O}$. 
Finally, $\pi:S\rightarrow \mathcal{P}(\Prop)$ is an interpretation of the atomic propositions $\Prop$  at the states.  

Intuitively, $PI(s)$ is the probability that an execution of the system starts at state $s$, and $PT(s,t)$ is
the probability that the state of the system at the next moment of time will be $t$, given that it is 
currently $s$.  The value $O_i(s)$ is the observation that agent $i$ makes when the system is in state $s$. 
(Below, in the context of interpreted systems, we treat the set of states $S$ as the states of the environment 
rather than as the set of global states. Agents' local  states will be derived from the observations.)  

Note that the first three components $(S,PI,PT)$ of a PO-DTMC form a standard discrete-time Markov chain. 
This gives rise to a probability space on runs in the usual way. 
A \emph{path}  in $M$ is a finite or infinite sequence $\rho = s_0 s_1 \ldots$ 
such that $PI(s_0) \neq 0$ and $PT(s_k, s_{k+1})>0 $ for all $k$ with $0 \leq k <|\rho|-1$. 
We write $\infinitepaths{M}$ for the set of all infinite paths of $M$, 
and, for $m \in \Nat$, write  $\finitepaths{M}{m}$ for the set of all finite paths $s_0 s_1 \ldots s_m$ with exactly $m$ transitions. 
Any finite path $\rho= s_0 s_1 \ldots s_m$  defines a set 
\beq\label{eq:basicCylinder}
\infinitepaths{M}\pathsover{\rho} = \{ \omega \in \infinitepaths{M}~|~ \omega[0\ldots m] = \rho \}
\eeq 
That is, $\infinitepaths{M}\pathsover{\rho}$ consists of all infinite paths which have $\rho$ as a prefix.

We now define a probability space $\Pspace(M) = (\infinitepaths{M}, \F, \mu )$ 
over the set $\infinitepaths{M}$ of all infinite paths of $M$. The $\sigma$-algebra $\F$ 
is defined to be the smallest $\sigma$-algebra over $\infinitepaths{M}$
that contains as basic sets all the sets $\infinitepaths{M} \pathsover{\rho}$ for  $\rho= s_0 s_1 \ldots s_m$ a finite path of $M$. 
For these basic sets, the function $\mu$ is defined by 
$$\mu(\infinitepaths{M}\pathsover{\rho}) = PI(s_0) \cdot PT(s_0,s_1) \cdot \ldots \cdot PT(s_{m-1},s_m)~.$$
The fact that $\mu$ can be extended to a measure on $\F$ 
is a non-trivial result of Kolmogorov for more general stochastic processes \cite{KemenySnell1}.

We may construct several different probabilistic interpreted systems from each PO-DTMC, 
depending on what agents remember of their observations.  
We consider two, one that assumes that agents have  perfect recall of all their observations, denoted $\spr$, 
and the other, denoted $\clk$, which assumes that agents are aware of the current time and their current observation. 
Recall that runs in an interpreted system map time to global states, 
consisting of a state of the environment and a local state for each agent. We interpret the states of the PO-DTMC $M$ as
states of the environment. To obtain a run, we also need to specify a local state for each agent
at each moment of time. We use the the observations to construct the local states. 

In the case of the {\em synchronous perfect recall semantics},  
given a path $\rho \in \infinitepaths M$, we obtain a run $\rho^{\spr}$ by defining the components at each time $m$ as follows. 
The environment state at time $m$ is $\rho_e^{\spr}(m)=\rho(m)$,  and the local state of agent $i$ at time $m$ 
is $\rho_i^{\spr}(m)=O_i(\rho(0))\ldots O_i(\rho(m))$. Intuitively, 
this local state assignment represents that the agent remembers all its past observations. 
We write $\R^\spr(M)$ for the set of runs of the form $\rho^{\spr}$ for $\rho\in \infinitepaths M$.
Note that this system is synchronous: if $r = \rho^\spr$ and $r' = \omega^\spr$ then for each agent $i$ 
and time $m \in \Nat$, if 
$r_i(m) = r'_i(m')$, then $O_i(\rho(0))\ldots O_i(\rho(m)) = O_i(\omega(0))\ldots O_i(\omega(m'))$, 
which implies $m = m'$. 

For the {\em clock semantics}, we  construct a run a  $\rho^{\clk}$ in which  again the environment state at time $m$ is $\rho_e^{\clk}(m)=\rho(m)$, and 
for agent $i$ we define the local state at time $m$ by $\rho^{\clk}(m) = (m,O_i(\rho(m))$. 
Intuitively, this says that the agent is aware of the clock value and its current observation. 
We write $\R^\clk(M)$ for the set of runs of the form $\rho^{\clk}$ for $\rho\in \infinitepaths M$ an infinite path of $M$. 
This system is also synchronous: if $r = \rho^\clk$ and $r' = \omega^\clk$ then for each agent $i$ 
and time $m \in \Nat$, if $r_i(m) = r'_i(m')$, then $(m, O_i(\rho(m))) = (m', O_i(\omega(m')))$, hence $m=m'$. 
In both cases of $x \in \{\spr,\clk\}$, if $T$ is a subset of $\infinitepaths M$, we 
write $T^x$ for $\{\rho^x~|~\rho\in T\}$. 

In both cases of $x \in \{\spr,\clk\}$, we have a one-to-one correspondence between the infinite paths
$\infinitepaths{M}$ and the runs $\R^x(M)$. We therefore 
can induce probability spaces $\Pspace^x(M)$ on $\R^x(M)$ from the probability space 
$\Pspace(M)$ on $\infinitepaths M$.   
As described above, the probability space $\Pspace^x(M)$ on runs moreover induces
a probability space $\prob^x_i(r,m)$ on the set of points considered possible by
each agent $i$ at each point $(r,m)$. 
 The PO-DTMC $M$ gives us an interpretation $\pi$ on its states, and we may derive from  this 
an interpretation $\pi^x$ on the points $(r,m)$ of  $\R^\spr(M)$ and  $\R^\clk(M)$ by defining $\pi^x(r,m) = \pi(r_e(m))$. 
Using the general construction defined above, we then obtain the probabilistic interpreted systems $\I^x(M) = \I(\R^x(M),\Pspace^x(M), \pi^x)$
for $x \in \{\spr,\clk\}$. 

It is necessary to establish the measurability of certain sets for the semantic definitions of the logics above to be complete. 
We now establish this for the systems $\I^\spr(M)$ and $\I^\clk(M)$. 

First, the general construction of $\I(\R,\Pspace, \pi)$, where $\Pspace = (\R, \F,\mu)$,  assumed that for all $m\in \Nat$ and all runs $r$, 
we have that $\R(\kset_i(r,m))\in \F$ is measurable and $\mu(\R(\kset_i(r,m)))>0$. 
The following lemma assures us that this is the case for both the perfect recall and the clock constructions. 
We write $\kset_i^x(r,m)$ for $\kset_i(r,m)$, relative to the set of runs $\R = \R^x(M)$. 

\begin{lemma}
Let $M$ be a finite PO-DTMC.  
Then for each $x\in \{\spr, \clk\}$, and run $r\in \R^x(M)$, 
the set $\R^x(M)(\kset_i^x(r,m))$ is measurable in $\Pspace^x(M)$ and has measure $>0$. 
\end{lemma}
\begin{proof}
For $x\in \{\spr, \clk\}$, define the relations $\sim^x$ on finite paths of $M$ as follows. 
For finite paths $\rho = s_0 \ldots s_m$ and $\omega = t_0 \ldots t_{k}$, 
define $\rho \sim_i^\clk \omega$ if $m=k$ and $O_i(s_m) = O_i(t_m)$, 
and  $\rho \sim_i^\spr \omega$ if $m=k$ and $O_i(s_j) = O_i(t_j)$ for all $j= 0\ldots m$. 
Note that in both cases, $\rho \sim_i^x \omega$ implies that $|\rho| = |\omega|$. 
Since the number of paths of length $m$ is finite, it follows that the number of $\omega$
such that $\rho \sim^x_i \omega$ is finite. If $r\in \R^x(M)$ then there exists $\rho \in \infinitepaths M$ such that 
$r = \rho^x$. Let $U = \{ \omega \in \finitepaths M m~|~ \omega \sim_i^x \rho[0\ldots m] \}$. 
It is straightforward from the definitions that 
\[ 
  \R^x(M)(\kset_i^x(\rho^x,m))= \bigcup_{\omega \in U} ~(\infinitepaths{M} \pathsover \omega)^x 
\]
The sets $\R^x(M)(\kset_i^x(\rho^x,m))$,  for $x\in \{\spr, \clk\}$,  are therefore each a finite union of measurable sets, hence measurable. 
The fact that their measure is $>0$ follows from the fact that both contain 
$ (\infinitepaths{M} \pathsover \rho)^x$, which has non-zero measure since $\rho$ is 
path of $M$. 
\end{proof}

Next, in giving the semantics of $\PLTLsK$, we require that sets of the form 
$\{(r',m')\in \kset^x_i(r,m)~|~\I^x(M),(r',m')\models\phi\}$ are measurable in 
$\Pspace^x_i(r,m)$. Similarly, in the semantics of  $\WMLOKP$, 
we require that sets of the form $T^x_\timeassgt(\phi) = \{r \in \R^x(M)~|~ \I^x(M),\timeassgt,r\models \phi\}$
are measurable, where $\timeassgt$ is a temporal assignment and $\phi$ is 
a formula of $\WMLOKP$. 
This is established in the following result. 

\begin{lemma} 
Let $M$ be a finite PO-DTMC and $x \in \{\spr, \clk\}$. 
For every set $S\subseteq \R(M)$ of runs of $M$ such that the semantic definitions above of 
$\PLTLsK$ and $\WMLOKP$ in $\I^x(M)$ refer to $\mu(S)$, the set $S$ is measurable in $\Pspace(M)$. 
\end{lemma} 

\begin{proof} 
We show this just for the logic $\WMLOKP$; for $\PLTLsK$ the result is a corollary 
derivable using using the fact that $\PLTLsK$ can be translated to $\WMLOKP$ (Proposition~\ref{prop:trans} 
below). We use structural induction on $\phi$ to show that $T^x_\timeassgt(\phi)$ is measurable for 
all temporal assignments $\timeassgt$. 

If $\phi= q(t)$  where $q$ is an atomic proposition and $t$ is a temporal variable, then 
note that for all paths $\rho \in \finitepaths{M}{\timeassgt(t)}$ 
and runs $r \in (\infinitepaths M \pathsover \rho)^x $, 
we have $\I^x(M), \timeassgt,r \models q(t)$ iff $q \in \pi(r(\timeassgt(t)))$ iff $q \in \pi(\rho(\timeassgt(t)))$. 
Thus, 
$$ T^x_\timeassgt(\phi) = \bigcup_{\rho \in \finitepaths{M}{\timeassgt(t)},~q\in \pi(\rho(\timeassgt(t)))} (\infinitepaths M \pathsover \rho)^x $$ 
which is a finite union of basic measurable sets, hence measurable. 

If $\phi= X(t)$  where $X$ is a set variable and $t$ is a temporal variable, then 
$ T^x_\timeassgt(\phi)$ is $\R(M)$ in case $\timeassgt(t) \in \timeassgt(X)$ and is 
$\emptyset$ otherwise. In either case, it is measurable. 
Similarly, if $\phi = t_1 <t_2$ then $T^x_\timeassgt(\phi)$ is either $\R(M)$ or $\emptyset$, and therefore measurable. 

If $\phi = \phi_1\land \phi_2$ then $T^x_\timeassgt(\phi) =  T^x_\timeassgt(\phi_1)\cap T^x _\timeassgt(\phi_2)$ is measurable, 
because, by the inductive hypothesis, $T^x_\timeassgt(\phi_1)$ and $ T^x _\timeassgt(\phi_2)$  are. 
Similarly,  $T^x_\timeassgt(\neg \phi) = \R^x(M) -T^x_\timeassgt(\phi)$ is measurable 
because  $T^x_\timeassgt(\phi)$ is measurable, by induction. 

For the quantifiers, we have 
$$ T^x_\timeassgt(\forall t( \phi)) = \bigcap_{m\in \Nat}   T^x_{\timeassgt[t\mapsto m]}( \phi)$$ 
and 
$$ T^x_\timeassgt(\forall X( \phi)) = \bigcap_{U\subset \Nat, ~ U ~finite }   T^x_{\timeassgt[X\mapsto U]}( \phi)~.$$ 
In either case, we have a countable intersection of sets that are measurable, by induction, 
hence these sets are also measurable.

For the knowledge formula $\phi = K_{i,t}\psi$, 
note that for $\rho\in \infinitepaths M$, 
if $\I^x(M), \timeassgt, \rho^x \models K_{i,t}\psi$, then
$$(\infinitepaths M \pathsover \rho[0\ldots  {\timeassgt(t)}])^x \subseteq T^x_\timeassgt(K_{i,t}\psi)~,$$ else 
$$(\infinitepaths M \pathsover \rho[0\ldots  {\timeassgt(t)}])^x \cap T^x_\timeassgt(K_{i,t}\psi)  = \emptyset~.$$ 
This is because, under both the clock and perfect recall semantics,
satisfaction of $K_{i,t}\psi$ depends only on the observations made to time $\timeassgt(t)$. 
Taking $$U = \{ \rho[0\ldots \timeassgt(t)]~|~\rho \in \infinitepaths M \text{ and }  \I^x(M), \timeassgt, \rho^x \models K_{i,t}\psi\}~,$$ 
we have  that 
\[ T^x_\timeassgt(K_{i,t}\psi) = \bigcup_{\rho \in P(i,\timeassgt, \psi)} (\infinitepaths M \pathsover \rho)^x \] 
is a finite union of basic measurable sets, hence measurable.

Finally, consider the probability  formulas $\phi = f(P_1, \ldots ,P_k) \bowtie c$, where $P_j = \prob_{i_j,t_j}(\psi_j)$ 
or $P_j = \prob(\psi_j)$. Let $m$ be the maximum of the values $\timeassgt(t_j)$ for $t_j$ a temporal 
variable in the former type of probability term. Note that the value 
$[ \prob(\psi_j)]_{\I,\timeassgt,r}$ is independent of the run $r$, and if $r$ and $r'$ are runs with 
$r[0,\ldots, m] = r'[0,\ldots,m]$ then $[ \prob_{i_j,t_j}(\psi_j)]_{\I^x(M),\timeassgt,r} =  [\prob_{i_j,t_j}(\psi_j)]_{\I^x(M),\timeassgt,r'}$. 
It follows also that when $r[0,\ldots, m] = r'[0,\ldots,m]$, we have $\I^x(M),\timeassgt,r \models  f(P_1, \ldots ,P_k) \bowtie c$ iff
 $\I^x(M),\timeassgt,r' \models  f(P_1, \ldots ,P_k) \bowtie c$. 
 Let $$U = \{\rho[0,\ldots, m]~|~ \rho \in \infinitepaths M,~ \I^x(M),\timeassgt, \rho^x \models  f(P_1, \ldots ,P_k) \bowtie c\} $$
 be the set of all prefixes of length $m$ of paths satisfying $ f(P_1, \ldots ,P_k) \bowtie c$. 
 The set $U$ is finite, and we have, by the above, that 
\[ T^x_\timeassgt(f(P_1, \ldots ,P_k) \bowtie c) = \bigcup_{\rho \in U}~   (\infinitepaths M \pathsover \rho)^x \] 
is a finite union of basic measurable sets, hence measurable. 
\end{proof}

\subsection{Discussion} \label{sec:discussion}

We have defined our logics to be quite expressive in the type of atomic 
probability assertions we have allowed, which involve polynomials of
probability expressions. In $\WMLOKP$, these expressions
may explicitly refer to different time points. Some  
existing logics of probability in the literature use a more restricted expressiveness, 
e.g., \cite{FaginHalpern} consider a logic that has only linear combinations
of probability expressions, and many logics \cite{BRS06,prismbook} allow only inequalities 
involving a single probability term. Here give some motivation to show that the richness 
we have allowed is natural and useful for applications. 

{\bf Polynomials:} There are several motivations for allowing polynomial combinations
of probability expressions. One, as noted in \cite{FaginHM90}, 
is that polynomials arise naturally from conditional probability. 
If we would like to include linear combinations of conditional probability expressions in the language, we find that this 
motivates a generalization to polynomial combinations of probability expressions. 
Consider the formula
$\prob(\phi_1|\psi_1) + \prob(\phi_2|\psi_2) \leq c$. Expanding out the definition of conditional probability, 
we have 
$$\frac{\prob(\phi_1\land \psi_1)}{\prob(\psi_1)} + \frac{\prob(\phi_2\land \psi_2)}{\prob(\psi_2)} \leq c~.$$ 
We see here that  there is a risk of division by zero that needs to be managed
in order for the semantics of this formula to be fully defined. 
One way to do so is to multiply out the denominators, resulting in the form 
$$\prob(\phi_1\land \psi_1)\cdot \prob(\psi_2) + \prob(\phi_2\land \psi_2)\cdot \prob(\psi_1) \leq c\cdot \prob(\psi_1)\cdot \prob(\psi_2)~$$ 
which is meaningful in all cases. (Should this not have the desired semantics in case  one of the $\prob(\psi_i)$ is zero, 
an additional formula can be added that handles this special case as desired.)  However, although we 
started with a linear probability expression, we now have multiplicative terms.  
This suggests that the appropriate way to add the expressiveness of conditional probability to the 
language is to admit atomic formulas that compare polynomial combinations of probability expressions. 

More generally, although it is less of relevance for purposes of model checking, 
and more of use for axiomatization of the logic, allowing polynomials also naturally enables 
familiar reasoning patterns  to be captured inside the logic. In particular, validities
such as $\prob(\phi_1 \lor \phi_2) = \prob(\phi_1) + \prob(\phi_2)$ 
when $\phi_1$ and $\phi_2$ are mutually exclusive and 
$\prob(\phi_1 \land \phi_2) = \prob(\phi_1) \cdot \prob(\phi_2)$ 
when $\phi_1$ and $\phi_2$ are independent show that both addition
and multiplication of probability terms arises naturally.

{\bf Mixed-time:} A second way in which our logics are rich  is in 
allowing probability atoms that refer to different moments of time. 
In $\PLTLsK$ this already the case because combinations such as 
$\prior_A(\phi) = \prob_A(\phi)$ are allowed, which refer to both the current 
time and to time $0$. The logic $\WMLOKP$ takes such temporal mixing
further by allowing reference to  time points explicitly named using time variables. 

Such temporal mixing is natural, since there are potential applications that require
this expressiveness. For example, in computer security, one often wants to 
say that the adversary $A$ does not learn anything about a secret from watching 
an exchange between two parties. However, it is often the case that 
the adversary  knows some prior distribution over the secrets. (For example, the 
secret may be a password, and choice of passwords by users are very non-uniform, 
with some passwords like `123456' having a very high probability.)  
This means that the simple assertion that the adversary does not know 
the secret, or that the adversary has a uniform distribution over the secret, 
does not capture the appropriate notion of security. 
Instead, as recognised already by Shannon in his work on secrecy \cite{shannon49}, 
we need to assert that the adversary's distribution over the 
secret has not changed as a result of its observations. 
This requires talking about the adversary's probability at two time points. 
For example, \cite{HLM11} capture an anonymity property  
by means of formulas using terms $\prior_A(\phi) = \prob_A(\phi)$. 

{\bf Mixed-time polynomials:}  Additionally, the logic of probability applied to formulas referring
to different times leads naturally  to polynomial combinations of 
probability terms, each referring to a different moment of time. 
For example, although $\WMLOP$  allows only formulas of the form 
$\prob(\phi(t_1, \ldots, t_n)) \bowtie c$, where the $t_i$ are time variables, the decision algorithm of \cite{BRS06} 
  uses the fact that, when $t_1< t_2 < \ldots < t_n$, the formula $\phi(t_1, \ldots, t_n)$ is equivalent to a formula 
of the form $\phi_1(t_1) \land \phi_2(t_2-t_1) \land \ldots \land \phi_{n}(t_n-t_{n-1}) \land  \phi_{n+1}(t_n)$, 
where the $\phi_i(t)$ are independent past-time formulas for $i = 1\ldots n$ and $\phi_{n+1}(t)$ is a  
future time formula. (This statement is closely related to Kamp's theorem \cite{kamp}.) This enables  $\prob(\phi(t_1, \ldots, t_n))$ to be expressed as a sum of products of terms of the form
 $\prob(\phi_i(u))$ where $\phi_i(u)$ has just a single free time variable $u$. 
 Thus, although mixed-time probability formulas are not directly expressible in the logic of \cite{BRS06}, 
 specific ones are implicitly expressible, and  the  extension is a mild one. It is worth remarking, 
 however, that the coefficients of the polynomial expansion of $\prob(\phi(t_1, \ldots, t_n))$ 
 are all positive, so we do not quite have arbitrary polynomials here. We return to this point below.

\section{Relating the logics} \label{sec:relations} 

The logic $\WMLOKP$ is very expressive, so it is not surprising that 
it can capture all of $\PLTLsK$. The following result makes this precise.

For the results below, it is convenient to add to the system a special agent $\bot$ that is blind,
and an agent $\top$ that has complete information about the state. 
In the context of PO-DTMC's 
these agents are obtained by taking the observation functions to satisfy 
$O_\bot(s) = O_\bot(t)$ and $O_\top(s) =s $ for all states $s,t$. 
 We write $\Box \phi$  for $K_{\bot, t} \phi$ where $t$ is any time
variable. This gives a \emph{universal modality}:  $\Box\phi$ says that $\phi$ holds on all runs.  
We write $[t\mapsto n]$ for the temporal assignment defined only on temporal variable $t$, and mapping this to $n$.

\begin{propn} \label{prop:trans} 
Let $M$ be a PO-DTMC with agent $\top$ and let $x \in \{\spr, \clk\}$. 
For every formula $\phi$ of  $\PLTLsK$, there exists a formula $\phi^*(t)$ of $\WMLOKP$
with $t$ the only free variable, 
such that $\I^x(M), (r,n) \models \phi$ iff $\I^x(M), [t\mapsto n], r \models \phi^*(t)$  for all runs $r$. 
\end{propn}

\begin{proof} 
The translation is defined by the following recursion: 
\[
\begin{split}
& p^*(t)  =   p(t)  \\
 & (\neg \phi)^*(t) = \neg \phi^*(t)  \\
 &  (\phi_1 \land \phi_2)^*(t) =  \phi^*_1(t) \land \phi^*_2(t) \\
&(X\phi)^*(t)  =  \exists u(u=t+1 \land \phi^*(u)) \\ 
& (K_i \phi)^*(t)  = K_{i,t}( \phi^*(t)),\\  
&(\phi_1 \until \phi_2)^*(t) =  \exists u\geq t( \phi^*_2(u) \land \forall v( t \leq v < u \rimp \phi_1^*(v)))\\
&(\prob_i(\phi))^*(t)  =  \prob_{i,t}(\phi^*(t)) \\
& (\prior_i(\phi))^*(t)  =   \prob_{i,0}(\phi^*(0)) \\ 
&(f(P_1, \ldots, P_k) \bowtie c)^*(t)  =  f(P^*_1(t), \ldots, P^*_k(t)) \bowtie c \\ 
\end{split} 
\]
Note that 
$u=v$ is definable as $\neg( u<v \lor v<u)$, that 
$u = t+1$ is definable as $ u >t  \land \forall v> t ~( u \leq v)$, and that $u = 0$ is definable as
$ \neg \exists  t( u = t+1)$. 
We can use $(A\phi)^*(t)  =   K_{\top, t} (\phi^*(t))$ to translate $A\phi$ in the perfect recall case. 
In case of the clock semantics, this translation loses the information about the initial state, 
which is required for correctness of the translation of $\prior_i(\phi)$. In this case, we
introduce, without loss of generality, new propositions $p_s$ for each state $s$, 
such that $p_s \in \pi_e(t) $ iff $s=t$, and take 
$$(A\phi)^*(t)  =   \bigwedge_{s\in S } (p_s(0) \rimp K_{\top, t} (p_s(0) \rimp \phi^*(t)))~.$$
\end{proof}

\newcommand{\obs}{\mathit{obs}}

With respect to the specific systems we derive from PO-DTMC's with respect to the
clock and perfect recall semantics, we are able to make some further statements 
that simplify the logic $\WMLOKP$ by eliminating some of the operators. 
These results are useful for the undecidability results that follow. 

For the following results, we note that, without loss of generality,
we may assume that a finite PO-DTMC comes equipped
with atomic propositions that encode the observations made by the agents. 
Specifically, when agent $i$ has possible observations 
$O_i(S) = \{o_{i,1} , \ldots , o_{i,k_i}\}$, 
we assume that there are atomic propositions 
$\obs_{i,j}$ for $i \in \Agt$  and $j= 1 \ldots k_i$ such that 
for all states $s$, 
we have 
$\obs_{i,j}\in \pi(s)$ iff $O_i(s) = o_{i,j}$.
Thus, $\obs_{i,j}(t)$ holds in a run just when agent $i$ makes observation $o_{i,j}$ at time $t$.

\begin{propn} \label{prop:elimkp-clock}
With respect to $\I^\clk(M)$ for a  finite PO-DTMC $M$, the  operators $K_{i,t}$ and $\prob_{i,t}$ 
can be eliminated using the universal operator $\Box$ and polynomial comparisons of universal probability terms $\prob(\psi)$, respectively. 
For simple probability formulas $\prob_{i,t}(\phi) \bowtie c$, only linear probability comparisons are required. 
\end{propn}

\begin{proof} 
The formula $$ \bigwedge_{j = i\ldots k_i} ( \obs_{i,j}(t) \rimp \Box( \obs_{i,j}(t)  \rimp \phi)$$ 
is easily seen to be equivalent to $K_{i,t}(\phi)$ in $\I^\clk(M)$.  
Similarly, $\prob_{i,t}(\phi) \bowtie c$
can be expressed as 
$$ \bigwedge_{j = i\ldots k_i} ( \obs_{i,j}(t) \rimp \prob (  \obs_{i,j}(t)  \land \phi ) \bowtie c \cdot \prob ( \obs_{i,j}(t) )) ~.$$
A similar transformation applies for more general agent probability comparisons, but we note that 
linear comparisons may transform to polynomial comparisons: similarly to the discussion of conditional probability in Section~\ref{sec:discussion}.    
\end{proof}

\begin{propn} \label{prop:qi} 
With respect to $\I^\spr(M)$ for a finite PO-DTMC $M$,
the  probability formulas $\prob_{i,t}(\phi) \bowtie c$ can be reduced to linear comparisons using only terms $\prob(\psi)$, provided 
second-order quantifying-in is permitted. Knowledge terms $K_{i,t}$ can be reduced to the universal modality $\Box$, 
provided   second-order quantifying-in is permitted for this modality. 
\end{propn}

\begin{proof} 
Define $\kappa_i(X_1, \ldots, X_{k_i}, t)$ to be the formula  $$\forall t' \leq t ( \bigwedge_{j= 1\ldots {k_i}} X_i(t') \dimp obs_{i,j}(t'))$$
Intuitively, this says that, up to time $t$, the second order variables $X_1, \ldots, X_k$ encode
the pattern of occurrence of observations of agent $i$ up to time $t$. 
The formula 
 $$ \forall X_1, \ldots X_{k_i}( \kappa_i(X_1, \ldots, X_{k_i}, t)  \rimp \Box( \kappa_i(X_1, \ldots, X_{k_i}, t)  \rimp \phi)$$ 
is easily seen to be equivalent to $K_{i,t}(\phi)$ in $\I^\clk(M)$.  
Similarly, $\prob_{i,t}(\phi) \bowtie c$
can be expressed as 
\[ 
\begin{split}
\forall X_1, \ldots X_{k_i}
(~~~  & \kappa_i(X_1, \ldots, X_{k_i}, t) \rimp \\
&   \prob (  \kappa_i(X_1, \ldots, X_{k_i}, t) \land \phi )  \bowtie c \cdot \prob ( \kappa_i(X_1, \ldots, X_{k_i}, t) ))~. 
\end{split}
\]
\end{proof} 

One might wonder whether the knowledge operators can be eliminated entirely using probability, 
treating $K_i \phi$ as $\prob_i(\phi) = 1$. This is indeed the case for formulas $\phi$ in $\PCTLK$. 
The essential reason is that because formulas of $\PCTLK$ depend at a point $(r,m)$ only on the run prefix $r[0 \ldots m]$, so 
the possibility that $\neg \phi$ holds on a non-empty set of measure zero does not occur. 

\begin{propn} \label{prop:probk} 
For all $\PCTLK$ formulas $\phi$ and PO-DTMC's $M$ and $x\in \{\clk, \spr\}$ we have $\I^x(M) \models K_i \phi \dimp \prob_i(\phi) = 1$. 
\end{propn}    

\begin{proof}
We write $T^x_\phi$ for $\{(r,m)\in \R^x(M)\times \Nat~|~\I^x(M),(r,m) \models \phi\}$
If $\I^{x}(M),(r,m)\models K_i\phi$ then 
$\kset^x_i(r,m) \subseteq T^x_\phi$. 
Hence by definition, $[\prob_{i}(\phi)]_{\I^x(M),(r,m)} =\mu^x_{r,m,i}(\kset^x_i(r,m)) =1$, 
so  $\I^{x}(M),(r,m)\models \prob_i(\phi) = 1$. 

Now assume that $\I^{x}(M),(r,m)\models \prob_i(\phi) = 1$. This implies that $\mu^x_{r,m,i}(\kset^x_i(r,m)-T^x_\phi)=0$. 
We have to show that then $Z_i=\kset_i(r,m)-T^X_\phi(r,m)$ is empty. Assume otherwise. If $(r',m) \in Z_i$
 then $\I^{x}(X),(r',m)\models \neg \phi$. Now in $\PCTLK$ we can show by induction that if $\I^x(M),(r,m)\models \psi$ 
 for any formula $\psi$ then  $\I^x(M),(r',m)\models \psi$ for any run $r'$ satisfying $r'[0\ldots m] = r[0\ldots m]$. 
 (We do not actually need to consider all subformulas of $\phi$, but can take the base case to be 
the subformulas of $\phi$ of the forms $A\psi$, $E\psi$, $K_j\psi$, $p$ and $f(P_1, \ldots,P_n)$ that are
not themselves a subformula of a larger formula of one of these types.  The claim is straightforward 
from the semantics for each of these types of formulas, and an easy induction over the 
boolean operators then yields the result for $\psi$.) 
Hence the set $\R_{r[0\dots m]}$ of runs having the initial prefix $r[0\dots m]$
is a subset of $Z_i$. Thus, $\mu^x_{r,m,i} (Z_i) \geq \mu^x_{r,m,i}(\R_{r[0\dots m]}) > 0$. This contradicts the assumption that $Z_i$ is a set of measure 0. 
\end{proof}

However, this is not the case for formulas $K_i \phi$ where $\phi$ is an LTL formula. Consider the following Markov Chain. 
\begin{figure}[h] 
\centerline{\includegraphics[height=2cm]{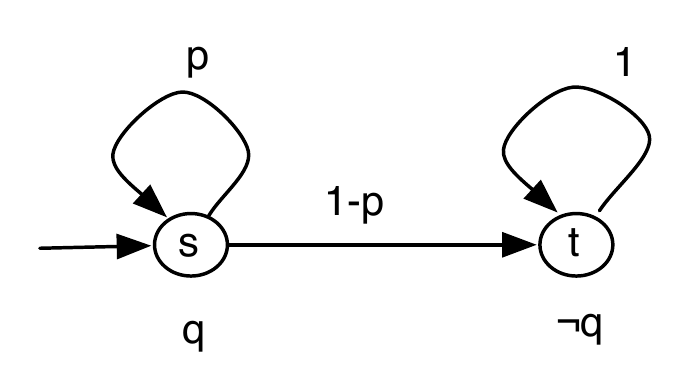}}
\end{figure} 
Here we have, at the initial state $s$, that $\neg K_i(F \neg q)$, because the agent considers it possible that 
always $q$ (this holds for all choices of observation functions). However, we have $\prob_i(F \neg q) = 1$, since the only run where $\neg q$ does not
eventually hold is the run that always remains at $s$. This run has probability zero. Hence we do not have 
that $\prob_i(F \neg q) = 1 \rimp K_i(F \neg q)$ is valid. (The converse form $K_i(\phi)  \rimp \prob_i(\phi) = 1$ 
does hold for all formulas $\phi$ of \PLTLsK, by the same argument as in the proof of Proposition~\ref{prop:probk}.)

\section{Undecidability Results} \label{sec:results} 

We can now state the main results of the paper concerning the problem of model checking formulas of (fragments of) the logics $\PLTLsK$
and $\WMLOKP$ in a PO-DTMC $M$, with respect to an epistemic  semantics $x\in \{\spr, \clk\}$. 
Using the results of Section~\ref{sec:relations}, we also obtain conclusions about extensions of $\WMLOP$ that 
do not refer to agent probability and knowledge. 

For a formula $\phi$  of $\PLTLsK$, we write $M\models^x \phi$, if $\I^x(M),(r,0)\models \phi$ for all runs $r\in \R^x(M)$.
In the case of $\WMLOKP$, we consider sentences, i.e., formulas without free variables, 
and  write $M\models^x \phi$, if $\I^x(M),\timeassgt, r\models \phi$ for all runs $r\in \R^x(M)$ and the empty 
assignment $\timeassgt$. 
The model checking problem is to determine, given a PO-DTMC $M$, a formula $\phi$,  and semantics $x \in \{\clk, \spr\}$, whether 
$M\models^x \phi$. 

\subsection{Background} 

For comparison with results below, it is worth stating a result from \cite{BRS06} concerning 
decidability of the fragment $\WMLOP$ of $\WMLOKP$ that omits knowledge operators $K_{i,t}$ and 
agent probability terms $\prob_{i,t}(\phi)$, restricts probability comparisons to the form $\prob(\phi) \bowtie c$, 
and prohibits second order quantification to cross into probability terms. 
Since the structure of agent's local states is irrelevant in this case, we write simply $\I(M)$ for the 
probabilistic interpreted system corresponding to a PO-DTMC $M$. 
To state the result, we define the \emph{parameterized} variant of a formula $\phi$ of 
$\WMLOP$ to be the formula $\phi_{x_1, \ldots, x_k}$, in which each subformula of the form $\prob(\psi) \bowtie c$
is replaced by a formula $\prob(\psi) \bowtie x_i$, with $x_i$ a fresh variable. 
We call the resulting formulas the \emph{parameterized formulas of $\WMLOP$.}  
For some $\alpha\in \Rat^k$, we can then recover the original formula $\phi$ as the instance $\phi_{\alpha}$
obtained from the parameterized variant $\phi_{x_1, \ldots, x_k}$  of $\phi$ by substituting $\alpha_i$ for  $x_i$ for 
each $i = 1\ldots k$.

\begin{theorem}[\cite{BRS06}] \label{thm:brs} 
For each parameterized sentence $\phi_{x_1, \ldots, x_k}$  of   $\WMLOP$,
one can compute for all $\epsilon >0$ a representation of a set $H_\phi\subset \Real^k$ of measure at most $\epsilon$, such that 
the problem of determining if $\I(M) \models \phi_\alpha$ is decidable for $\alpha \in \Rat \setminus H_\phi$.
\end{theorem} 

Intuitively, the complement of $H_\phi$ contains the points that are bounded away from 
limit points of the Markov chain, and comparisons can be decided using 
convergence properties.

The reason for excluding the set $H_\phi$ is that the limit point cases seem to 
require a resolution of problems related to the \emph{Skolem problem} concerning zeros of linear recurrences \cite{Skolem34}. 
A sequence of real numbers $\{u_n\}$ is called a linear recurrence sequence (LRS) of order $k$ if 
there exist $a_1, \ldots a_k$ with $a_k \neq 0$ such that  
for all $m \geq 1$, 
\[ u_{k+m} = a_1u_{k+m-1} + a_2u_{k+m-2}+\dotsb+ a_ku_m~.\]
We consider the following decision problems associated with a LRS $\{u_n\}$. 
\be
\item
{\bf Skolem problem.} Does there exist $n$ such that $u_n =0$?
\item
{\bf Positivity problem.} Is it the case that for all $n$, $u_n \geq 0$?
\item 
{\bf Ultimate positivity problem.} Does a positive integer $N$ exist such that for all $n \geq N$, $u_n \geq 0$?
\ee
We will deal with sequences with rational entries. By clearing denominators the rational version of the above problems can be shown to be polynomially equivalent to 
similar problems stated using sequences with integer entries. 
There has been a significant amount of work on these problems \cite{Igorbook}, but they have stood unresolved since the 1930's. 
To date, only low order versions of these problems have been shown to be decidable \cite{Halava05,OW14,TMS84}. 

The above problems have an equivalent matrix formulation. A proof of the following can be found in \cite{Halava05}.
\begin{lemma}
For a sequence $u_0, u_1,\dotsc,$ the following are equivalent. 
\be 
\item
$\{u_n\}$ is a rational LRS.  
\item
For $n\geq 1$, $u_n = (A^n)_{1k}$ for a square matrix $A$ with rational entries. 
\item
For $n\geq 1$, $u_n= \ve{v}^TA^n\ve{w}$ where $A$ is a square matrix, and $\ve{v}$ and $\ve{w}$ are vectors with entries from $\{0,1\}$.
\ee 
\end{lemma}
In the usual formulation of the Skolem, positivity and ultimate positivity problems,  
the associated matrices $A$ may contain  negative numbers, and numbers not in $[0,1]$,
so are not stochastic matrices.  
However, \cite{AAOW15} show that these problems can be reduced to a decision problem stated with 
respect to stochastic matrices:  
\begin{lemma} \label{red-A-B} 
Given an integer $k\times k$ matrix $A$, one can compute a $k'\times k'$ stochastic matrix $B$, 
a length $k'$ stochastic vector $\ve{v}$, a length $k'$ vector $\ve{w}= (0, \ldots, 0,1)$ and a constant 
$c$ such that $(A^n)_{1,k} = 0$ ($(A^n)_{1,k} >0$) iff $\ve{v}^T B^n \ve{w} = c$ (respectively, $\ve{v}^T B^n \ve{w} >c$).  
\end{lemma} 
As noted in  \cite{AAOW15}, it follows that the logic $\WMLOP$ is 
able to express the Skolem and positivity problems, 
using model checking questions of the form 
$$ M \models \exists t (\prob(p(t)) = c)$$ 
and
$$ M \models \exists t (\prob(p(t)) > c)$$ 
for $c$ a nonzero constant and $M$ a DTMC. 
(The ultimate positivity problem can also be expressed.) 
It is worth noting that in the special case of the constant $c=0$, 
these model checking questions \emph{are} decidable, 
as shown in \cite{BRS06}. (Essentially, in this case the problems 
reduce to  graph reachability problems, and the specific 
probabilities in $M$ are irrelevant.) The transformation from 
arbitrary matrices $A$ to stochastic matrices $B$ in Lemma~\ref{red-A-B}
requires that the constant $0$ of the Skolem problem be replaced by a non-zero constant $c$.  

The above model checking problems of the quantified logic $\WMLOP$ can be seen to be already expressible in the propositional logic $\PLTLsK$, as 
the problems 
\begin{subequations}\label{eq:formula_Skolem}
\begin{equation*}
M' \models^\clk \mathbf{AF} (\pr_i(p) =c) 
\end{equation*}
\begin{equation*}
M' \models^\clk \mathbf{AF} (\pr_i(p) > c)  
\end{equation*}
\begin{equation*}
M' \models^\clk\mathbf{AF} \mathbf{AG}(\pr_i(p) > c )
\end{equation*}
\end{subequations}
where we obtain the PO-DTMC $M'$ from the DTMC $M$ 
by defining $O_i(s) = \bot$ for all states $s$. 
That is, agent $i$ is blind, so considers all states 
reachable at time $n$ to be possible at time $n$. 
(We remark that this implies that all the operators $A$ can be exchanged with $E$ 
without change of meaning of the formulas.) 
It follows, that with respect to clock semantics, 
a resolution of the decidability of model checking even these simple 
formulas of $\PLTLsK$ for \emph{all} $c\in [0,1]$ would imply a resolution of the Skolem problem. 
In view of the effort already invested in the Skolem problem, this is 
likely to be highly nontrivial.

\subsection{Perfect Recall Semantics} 

Model checking with respect to the perfect recall semantics is undecidable, even with respect to a very simple 
fixed formula of the logic $\PLTLsK$, as shown by the following result. 

\begin{theorem} \label{thm:pr} 
The problem of determining, given a PO-DTMC $M$, if  $M \models^\spr EF (\prob_i(p)> c)$, for $p$ an atomic proposition, is 
undecidable. 
\end{theorem}

\begin{proof} 
By reduction from the emptiness problem for probabilistic finite automata \cite{Paz71}. 
Intuitively, the proof sets up an association between words in the matrix semigroup and sequences of observations of the agent. 

A probabilistic finite automaton is a tuple $\PFA = (\PAS,\PAalph,\PAinit,\PAtrans, \PAfinal, \lambda)$, 
where   $\PAS$ is a finite set of states, $\PAalph$ is a finite alphabet, $\PAinit:\PAS \rightarrow [0,1]$ is a probability distribution 
over states, representing the initial distribution, $\PAtrans: \PAalph \rightarrow (\PAS\times\PAS \rightarrow[0,1])$ 
associates a transition probability matrix $\PAtrans(a)$ with each letter $a\in \PAalph$, component $\PAfinal \subseteq \PAS$ is a 
set of \emph{final} states, and 
$\lambda \in (0,1)$ is a rational number. Each matrix $\PAtrans(a)$ satisfies $\sum_{t\in S} \PAtrans(a)(s,t) = 1$ for all $s\in \PAS$. 
Let $v_\PAfinal$ be the column vector indexed by $\PAS$  with $v_\PAfinal(s) = 0$ if $s\not \in F$ and 
$v_\PAfinal(s) = 1$ if $s \in F$. Treating $\PAinit$ as a row vector, 
for each word $w = a_1 \ldots a_n \in \PAalph^+$, define $f(w) = \PAinit \PAtrans(a_1) \ldots  \PAtrans(a_n) v_\PAfinal$. 
The language accepted by the automaton is defined to be ${\cal L}(\PFA) =  \{w\in \PAalph^+~|~ f(w) > \lambda\}$.  
The emptiness problem for probabilistic finite automata is then, given a probabilistic finite automaton $\PFA$, 
to determine if the language ${\cal L}(\PFA)$ is empty.  
This problem is known to be undecidable \cite{Paz71,CondonLipton89}. 

Given a probabilistic finite automaton  $\PFA$, we construct 
an interpreted finite PO-DTCM $M_\PFA$ for a single agent (called $i$ rather than $1$ to avoid confusion with other numbers) such that $M_\PFA \models^\spr EF (\prob_i(p)> \lambda)$
iff $\PFA$ is nonempty. This system is defined as follows. We let $N = |\PAalph|$, 
\be
\item $S = \PAS \times \PAalph$, 
\item $PI(q,a) = \mu_0(q)/N$, 
\item  $PT((q,a), (q',b)) = \PAtrans(b)(q,q')/N$ 
\item $O_i((q,a)) = a$ 
\item $p \in \pi((q,a))$ iff  $q\in \PAfinal$.
\ee 
Note that $\sum_{(q,a) \in S} PI((q,a)) = \sum_{a\in \PAalph}\sum_{q\in \PAS} \mu_0(q)/N = \sum_{a\in \PAalph} 1/N = 1$, 
so $PI$ is in fact a distribution. Similarly, for each $(q,a) \in S$, we have 
$\sum_{(q',b)\in S} PT((q,a), (q',b)) = \sum_{b\in \PAalph}\sum_{q'\in \PAS}\PAtrans(b)(q,q')/N =  \sum_{b\in \PAalph} 1/N = 1$, 
so $PT$ is in fact a stochastic matrix. 

\newcommand{\B}{{\cal B}}

Note that for each $w = a_1\ldots a_n\in \PAalph^*$ and $a\in \PAalph$, we get a row vector 
$\mu_{aw} = \mu_0 \PAtrans(a_1) \ldots \PAtrans(a_n)$ with $\sum_{q\in \PAS} \mu_{aw}(q) =1$, 
which can be understood as a distribution on $\PAS$. For each run $r \in \R^\spr(M_\PFA)$
and $m\geq 0$, we have  that that agent $i$'s local state $r_i[0\ldots m]$ at $(r,m)$ is a word in $\PAalph^+$. 
Let $\B(q,m)$ be the set of runs $r\in  \R^\spr(M_\PFA)$ in which $r_e(m) = (q,a)$ for some $a\in \PAalph$. 
We claim the following about the 
probability measures $\mu_{r,m,i}$ in the probabilistic interpreted system $\I^\spr(M_\PFA)$, for each point $(r,m)$ and $q\in \PAS$: 
$$ \mu_{r,m,i}(\K_i(r,m)(\B(q,m))) = (\PAinit \PAtrans(r_i(1)) \ldots \PAtrans(r_i(m)))(q)~.$$
It is immediate from this that $\I^spr(M_\PFA) ,(r,m) \models \prob_i(p) = c$ where $c = f(r[1\ldots m])$, and the 
desired result follows.

For the proof of the claim, 
let $\mu$ be the measure on $\R^\spr(M_\PFA)$ in $\Pspace(M_\PFA)$. 
Note first that $\mu(\R^\spr(\K_i(r,m))) = 1/N^{m+1}$ for all runs $r$. 
The proof of this is by induction on $m$. Let $r_i(m) = a_0 \ldots a_m$. 
For $m=0$, we have $\mu(\R^\spr(\K_i(r,m)))  = \sum_{q\in \PAS} \PAinit(q)/N = 1/N$. 
For $m = k+1$, 
\begin{align*} 
& \mu(\R^\spr(\K_i(r,m))) \\
&  = \mu( \{ r'|~ r'_i[0\ldots k] = a_0\ldots a_k,~ r_e(k) = (q,a_k), ~r_e(k+1) = (q',a_{k+1})\}) \\ 
& = \sum_{q\in \PAS}\sum_{q'\in \PAS} \mu( \{ r'|~ r'_i[0\ldots k] = a_0\ldots a_k,~ r_e(k) = (q,a_k)\}) \cdot PT((q,a_k),(q',a_{k+1}) \\ 
& = \sum_{q\in \PAS} \mu( \{ r'|~ r'_i[0\ldots k] = a_0\ldots a_k,~ r_e(k) = (q,a_k)\}) \cdot \sum_{q'\in \PAS} PT((q,a_k),(q',a_{k+1}) \\
& = \sum_{q\in \PAS} \mu( \{ r'|~ r'_i[0\ldots k] = a_0\ldots a_k,~ r_e(k) = (q,a_k)\}) \cdot \sum_{q'\in \PAS} \PAtrans(a_{k+1})(q,q')/N\\
& = \sum_{q\in \PAS} \mu( \{ r'|~ r'_i[0\ldots k] = a_0\ldots a_k,~ r_e(k) = (q,a_k)\}) \cdot 1/N\\
& =  \mu( \{ r'|~ r'_i[0\ldots k] = a_0\ldots a_k\}) \cdot 1/N\\
& = 1/N^{m+1} 
\end{align*} 
with the last step by induction. 
 
For the main claim, we again proceed by induction on $m$. 
For the case $m = 0$, let $r_e(0)  = (q,a)$ 
we have $(r,m) \sim_i (r'm')$ iff $r'_e(0) = (q',a)$ for some $q'$, 
so $\K_i(r,m) =  \{r'~|~ r'_e(0) = (q',r_i(0))\text{ for some $q'\in \PAS$}\}$. 
Hence 
\begin{align*} 
& \mu_{r,m,i}(\K_i(r,m)(\B(q,m))) \\
&  =
\mu( \R^\spr(\K_i(r,m)(\B(q,m)))) / \mu(\R^\spr(\K_i(r,m))) \\
 & = \mu( \{r'~|~ r'_e(0) = (q,r_i(0))\}) / \mu( \{r'~|~ r'_e(0) = (q',r_i(0)) \text{ for some $q'\in \PAS$}\})  \\
 & =  PI(q,r_i(0)) / PI( \{(q',r_i(0))  ~|~q'\in \PAS\})\\ 
& = (\PAinit(q)/N)/(1/N) \\
& = \PAinit(q)
\end{align*}
as required. 

For $m = k+1$, note 
\begin{align*} 
& \mu_{r,k+1,i}(\K_i(r,k+1)(\B(q,k+1))) \\ 
&  =
\mu( \R^\spr(\K_i(r,k+1)(\B(q,k+1)))) / \mu(\R^\spr(\K_i(r,k+1))) \\
 & = N^{k+2} \cdot \mu( \R^\spr(\K_i(r,k+1)(\B(q,k+1)))
\end{align*}
by the above fact. 
Now, with $r_i[0\ldots k+1] = a_0 \ldots a_{k+1}$, we have 
\begin{align*}
& \mu( \R^\spr(\K_i(r,k+1)(\B(q,k+1))) \\
 & = 
\mu( \{r'~|r'_i[0\ldots k+1] = a_0 \ldots a_{k} , ~ r'_e(k+1) = (q,a_{k+1}) \}) \\ 
& = \sum_{q'\in \PAS}  \mu( \{r'~|r'_i[0\ldots k] = a_0\ldots a_k, r'_e(k) = (q', a_k), ~ r'_e(k+1)= (q,a_{k+1}) \}) \\ 
& = \sum_{q'\in \PAS}  \mu( \{r'~|r'_i[0\ldots k] = a_0\ldots a_k, r'_e(k) = (q', a_k)\}) \cdot PT((q',a_k), (q,a_{k+1})) \\ 
& = \sum_{q'\in \PAS}  \mu( \{r'~|r'_i[0\ldots k] = a_0\ldots a_k, r'_e(k) = (q', a_k)\}) \cdot \PAtrans(a_{k+1})(q',q)/N \\ 
& = \sum_{q'\in \PAS}  \mu( \{r'~|r'_i[0\ldots k] = a_0\ldots a_k, r'_e(k) = (q', a_k)\}) \cdot \PAtrans(a_{k+1})(q',q)/N \\ 
& = \sum_{q'\in \PAS}  (\PAinit \PAtrans(a_1) \ldots \PAtrans(a_m))(q') \cdot \PAtrans(a_{k+1})(q',q)/N^{k+1}N  \quad\quad\text{(by induction)} \\ 
& = \PAinit \PAtrans(a_1) \ldots \PAtrans(a_m) \PAtrans(a_{k+1})(q)/N^{k+2}   \\ 
\end{align*} 
Combining this with the previous equation, we have 
$\mu_{r,k+1,i}(\K_i(r,k+1)(\B(q,k+1)))  =\PAinit \PAtrans(a_1) \ldots \PAtrans(a_m) \PAtrans(a_{k+1})(q)$ 
as required.  
\end{proof} 

We remark that this result stands in contrast to the situation for model checking the 
logic of knowledge and time. Write $\mathbf{CLTL^*K}$ for the logic  obtained
from $\PLTLsK$ by omitting the probability comparison atoms $f(P_1, \ldots,P_k) \bowtie c$. 
Model checking the logic $\mathbf{CLTL^*K}$ with respect to perfect recall, 
i.e., deciding $M \models^\spr \phi$ for $M$ a PO-DTMC and $\phi$ a formula 
is decidable \cite{MeydenShilov}. (Here, for the semantic structures $M$, 
it suffices to replace the initial distribution $PI$ in $M$ by the set $I = \{s\in S~|~ PI(s)>0\}$, 
and replace the transition distribution function $PT$ in $M$ by the 
relation $R$ of possibility of transitions between states defined by $sRt$ if $PT(s,t)>0$.
The results in \cite{MeydenShilov} use linear time temporal logic as a basis, 
but, as noted in \cite{MeydenWong03}, the modality $A$ of the branching time logic $CTL^*$ can 
be understood as a special case of a knowledge modality: see Proposition~\ref{prop:trans}.) 
 
For probabilistic automata the minimum size of the state space giving undecidability directly stated in the literature appears to be 25 \cite{Hirvensalo06}.
We remark that the proof of Theorem~\ref{thm:pr}  can also be done by reduction of the following matrix semigroup problem: 
{\em given a finite set of matrices of order $n$, generating  a matrix semigroup $S$, determine whether there is $M\in S$ such that $(M)_{1n} =0$} \cite{Halava97}. 
The case of $k$ generators of size $n\times n$ can be reduced to probabilistic automata with $2kn+1$ states. 
Recent results on the matrix semigroup problem are given in \cite{CHHN14}.

Huang et al \cite{HSZ12} have previously used a reduction from probabilistic automata to 
show undecidability of an probabilistic epistemic logic with respect to perfect recall. 
Compared to our simple CTL temporal operators, their logic uses more expressive setting of 
alternating temporal logic  operators. 

\subsection{Clock Semantics} 

The undecidability of the perfect recall semantics for such simple formulas suggests that
we weaken the epistemic semantics to the clock case. The combination of 
the translation from $\PLTLsK$ to $\WMLOKP$ (Proposition~\ref{prop:trans}) 
and Theorem~\ref{thm:brs} then enables some cases of $\PLTLsK$ to be decided. 
We do not obtain a full decidability result, however, since we face the problem that, with respect to the clock semantics, 
the formula  $AF (\prob_i(p)=c)$ can express the Skolem problem, so resolving 
its decidability is a very difficult problem. Rather than attempt to resolve this question, 
we consider here just how much extra expressiveness is required
over the logic of Theorem~\ref{thm:brs} for us to obtain a definitive {\em undecidability} result, 
instead of a decidability result with some excluded and unresolved cases. 

Note that one of the restrictions on $\WMLOP$ used in Theorem~\ref{thm:brs} is 
that second order quantification should not cross into probability terms. It turns out 
that this restriction is essential, as shown by the following result. 

\begin{theorem} \label{thm:quantinundec}
It is is undecidable, given a PO-DTMC $M$ and a formula $\phi$ of $\WMLOKP$ with linear combinations of probability 
terms $\prob(\phi)$ and quantifying-in of second-order quantifiers, whether $M\models \phi$.  
\end{theorem}

\begin{proof} 
This follows from the fact that, using second order quantifying-in, we can express 
perfect recall (Proposition~\ref{prop:qi}), and the undecidability of model checking perfect recall (Theorem~\ref{thm:pr}). 
\end{proof} 

Note that the result refers to $\models$ rather than $\models^\clk$, since epistemic operators are
not required. This is really a result about a generalization of $\WMLOP$.
One of the other restrictions in Theorem~\ref{thm:brs} 
is that only simple probability comparisons of the form 
$\prob(\phi) \bowtie c$ are permitted. 
More general comparisons of  probability terms are needed in applications (see discussion in Section~\ref{sec:discussion}), so it is of interest to study their 
impact on decidability. Unfortunately, it turns out to be quite negative. Even the 
simple case of mixed time polynomial atomic probability formulas is enough for undecidability. 

\begin{theorem} \label{thm:mixedtime} 
There exists  a fixed PO-DTMC $M$ with 4 states such that it is undecidable, 
given a mixed-time polynomial  atomic probability formula  $\psi$, whether $M \models \psi$.   
\end{theorem}

\begin{proof} 
By reduction from Hilbert's tenth problem, i.e., the problem 
of determining whether a polynomial with integer coefficients 
has solutions in the natural numbers. This was shown to be 
undecidable by Matiyasevich \cite{Matiyasevich}. 

We show that we can find a stochastic matrix $M$ 
and a stochastic vector $\ve{f}$ 
such that for each function $f(t) = t\cdot \lambda^{t}$ and $f(t) = \lambda^{t}$ with 
$\lambda = 1/2$, 
there 
is a rational vector $\ve{g}$
such that  $f(t) = \ve{f}^T M^t\ve{g}$. Given a polynomial $p(n_1, \ldots, n_k)$, we can construct a variant polynomial 
$q'$ over a larger set of variables, such that an appropriate 
substitution of such functions $t_i\cdot \lambda^{t_i}$ and $\lambda^{t_i}$, 
for the $n_i$ and the additional variables yields an 
expression 
$\lambda^{d_1 t_1 + \ldots + d_k t_k}\cdot p(t_1, \ldots, t_k)$, where the $d_i$ are constants. 
This has a zero in the $t_1\ldots t_n$ iff  $p(x_1, \ldots, x_n)$ has a zero. 
It follows 
 that  mixed-time polynomial  atomic probability formulas can express Hilbert's tenth problem.

 Consider the following matrix 
\beq\label{eq:matrix_hilbert}
A = 
\begin{pmatrix}
\frac{1}{2} & 0 & 0 & \frac{1}{2} \\
 \frac{1}{4} & \frac{1}{2} & 0 & \frac{1}{4} \\
\frac{1}{3} & \frac{1}{3} & \frac{1}{3} & 0 \\
0 & 0 & 0 & 1
\end{pmatrix}
\eeq
It is easily checked that 
$A$ is stochastic matrix. Expanding by the last  row, the characteristic (and minimal) polynomial of $A$ is given by
\[ f(x) = (x-1)(x-\frac{1}{2})^2(x-\frac{1}{3}) \]
Using Perron's equation \cite{Romanovsky}, which provides a way to express the powers of a matrix using a polynomial,  
we have 
\beq\label{eq:ex_hilbert1}
\begin{split}
A^n = & ~6 (A -\frac{1}{2}I)^2(A-\frac{1}{3} I) ~- ~12 (\frac{1}{2})^n (A- I) (A-\frac{1}{3} I)~ - \\
&12 [(\frac{1}{2})^{n-1} n - 4(\frac{1}{2})^n](A- I)(A-\frac{1}{2}I)(A-\frac{1}{3} I) ~-~54 (\frac{1}{3})^n (A-I)(A-\frac{1}{2}I)^2 \\
\end{split}
\eeq
(The reader may verify the correctness of this equation using the fact (see e.g., \cite{Gantmacher}, Ch. V) that if, for a polynomial $p(x)$ and
a matrix $B$, we have, for each eigenvalue $\lambda$ of $B$, of multiplicity $k$, we have 
$p^{(j)}(\lambda) = 0$ for all $j= 0 \ldots k-1$ (where  $p^{(j)}$ is the $j$-th derivative), 
then $p(B) = 0$.)

Define the vectors 
$$
\ve{f} = \begin{pmatrix}
1/4 \\ 1/4 \\ 1/4 \\ 1/4
\end{pmatrix} \quad 
\ve{g} = \begin{pmatrix}8/3 \\ 8/3 \\  -8 \\ 0 \end{pmatrix} 
\quad 
\ve{g}' = 2 (A-\frac{1}{2} I) \ve{g} = \begin{pmatrix} 0 \\ 4/3 \\  8/3 \\ 0 \end{pmatrix} 
\quad \ve{g}'' = 2 \ve{g}' + \ve{g} = \begin{pmatrix} 8/3 \\ 0 \\  - 8/3 \\ 0 \end{pmatrix} 
$$
Then $\ve{g}'$ is an eigenvector  of $A$ with eigenvalue $1/2$, and
$(A-\frac{1}{2})^2\ve{g} = 0$. Setting $k(x) = (x-1)(x- \frac{1}{3})$, 
we also have $\ve{f}^T k(A) \ve{g} = 0$
and $\ve{f}^T k(A) \ve{g}' = -\frac{1}{12}$. Using equation~\ref{eq:ex_hilbert1}, 
we obtain
\beq \label{eq:ex_hilbert2}
\ve{f}^T A^n \ve{g} = (\frac{1}{2})^n[n-2] \quad \quad \ve{f}^T A^n \ve{g}' = (\frac{1}{2})^n \quad\quad  \ve{f}^T A^n \ve{g}'' = (\frac{1}{2})^n \cdot n 
\eeq
for all $n\geq 0$. 

We construct a PO-DTMC $M$ for a single agent (called 1), from $\ve{f}$ and $A$, by taking states $S= \{s_1,s_2,s_3,s_4\}$ corresponding
to the four dimensions, the initial distribution $PI$ to be given by $\ve{f}$, the transition probabilities $PT$ to be given by $A$, 
and the observation function defined by $\obs_1(s) = \bot$ for all $s\in S$. (That is, the agent is blind.) 
We define $\pi$ to interpret the propositions $p_1, \ldots , p_4$ by $p_i \in \pi(s_j)$ iff $i=j$. That is, each $p_i$ holds at state $s_i$ only. 
Let $\I = \I(M)$. 

Corresponding to the vector $\ve{g}'$, we define the probability term 
$$x(t) = \frac{4}{3} \prob_{1,t}( p_2) +  \frac{8}{3} \prob_{1,t}( p_3) $$ 
and corresponding to the vector  $\ve{g}''$, 
we define the probability term 
$$w(t) = \frac{8}{3} \prob_{1,t}( p_1) - \frac{8}{3} \prob_{1,t}( p_3) ~.$$
Then for all runs $r$ of $\I$ and temporal assignments $\timeassgt$ 
such that $\timeassgt(t_i) = n_i$, we have 
$$[ x(t_i) ]_{\I,\timeassgt,r} = (\frac{1}{2})^{n_i} \quad\text{and} \quad  [ w(t_i) ]_{\I,\timeassgt,r} = (\frac{1}{2})^{n_i} \cdot n_i~.$$ 

Now,  any diophantine equation 
\[p(n_1,n_2,\dotsc, n_k) = \sum_{i_1,\dotsc, i_k}  a_{i_1,\dotsc, i_k}n_1^{i_1}n_2^{i_2} \dotsc n_k^{i_k} = 0\]
can be expressed as follows. First let $d_i$ be the degree of $P$ in the variable $n_i$. Corresponding to the monomial $m_{i_1, \dotsc, i_k} = n_1^{i_1}n_2^{i_2} \dotsc n_k^{i_k}$, define the 
probability term  
$$z_{i_1, \dotsc, i_k} = w(t_1)^{i_1} \dotsc  w(t_k)^{i_k} x(t_1)^{d_1-i_1} \dotsc  x(t_k)^{d_k-i_k}~.$$ 
Using these terms, we obtain the probability term 
$$Z(t_1, \dotsc, t_k) = \sum_{i_1,\dotsc, i_k}  a_{i_1,\dotsc, i_k}z_{i_1, \dotsc, i_k}~.$$ 
Now,   for the temporal assignment $\timeassgt$ with $\timeassgt(t_i) = n_i$ for $i = 1 \ldots k$,
 we have 
\[  [Z(t_1, \dotsc, t_k)]_{\I,\timeassgt,r}   =  \left (\frac{1}{2}\right)^{n_1d_1+\dotsb +n_kd_k} \cdot \sum_{i_1,\dotsc, i_k}  a_{i_1,\dotsc, i_k}n_1^{i_1}n_2^{i_2} \dotsc n_k^{i_k} \]
Hence,  we have 
$[Z(t_1, \dotsc, t_k)]_{\I, \timeassgt,r} = 0 $ iff $P(n_1,n_2,\dotsc, n_k) =0$. 
It follows that $\I \models \exists t_1 \ldots t_k( Z(t_1, \dotsc, t_k) = 0)$ iff $P$ has a solution in the natural numbers. 
\end{proof} 

We remark that the possibility of encoding Hilbert's tenth problem is not 
immediate from the fact that we are dealing with polynomials, since
our polynomials are over \emph{rational}  values generated in a very specific way
from Markov chains, rather than arbitrary integers. Indeed, there are decidable logics containing polynomials, 
such as the theory of real closed fields \cite{Tar1}. 

As noted in Section~\ref{sec:discussion}, formulas 
(allowed by Theorem~\ref{thm:brs}) of the form $\prob(\phi(t_1, \ldots, t_n)) \bowtie c$ can be written as a polynomial of probability expressions, so it is natural 
to ask whether such formulas also suffice  to make the logic undecidable. 
This does not seem to be the case:  the polynomials involved have
only positive coefficients. Since Hilbert's tenth problem is trivially decidable for 
polynomials with only positive coefficients, our proof does not apply to this case.

\section{ Conclusion} \label{sec:concl} 

Our results have by no means resolved Skolem's problem, which remains an apparent barrier
to resolving the gap between the decidability results of \cite{BRS06} and the 
undecidability results of the present paper. 

However, in work to be presented elsewhere, 
we show that the results of \cite{BRS06} can be extended both by reducing the set $H_\phi$ of cases that 
needs to be excluded to obtain decidability, as well as enhancing the 
expressiveness  to cover epistemic probabilistic terms of the form $\prob_i(\phi)$, interpreted with 
respect to the clock semantics. \\[5pt]

\noindent
{\bf Acknowledgment:} The authors thank Igor Shparlinski and Min Sah for illuminating discussions on Skolem's problem, and Xiaowei Huang 
for helpful comments on a draft of the paper.

\bibliographystyle{alpha} 
\bibliography{pxn}

\end{document}